\documentclass[english,fleqn]{article}
\usepackage[utf8]{inputenc}
\usepackage{babel}

\usepackage{amsmath,amsfonts,amssymb, amsthm}
\usepackage{commath}
\usepackage{graphicx}     

\usepackage{multicol}
\newtheorem{theorem}{Theorem}
\newtheorem{lemma}{Lemma}

\newtheorem{corollary}{Corollary}

\title{The Mueller matrix cone and its application to filtering}
\author{Tim Zander\thanks{ Karlsruhe Institute of Technology,
    Vision and Fusion Laboratory (IES),
    Karlsruhe}
  \\ tim.zander@kit.edu  
  \and
  Jürgen Beyerer\footnotemark[1]
  \thanks{Fraunhofer Institute of Optronics, System Technologies and Image Exploitation (IOSB), Karlsruhe}
  \\  juergen.beyerer@iosb.fraunhofer.de
}
\usepackage[backend=biber]{biblatex}
\usepackage{hyperref}

\begin{document}
\newpage
\maketitle
\begin{abstract}
  We show that there is an isometry between the real ambient space of all Mueller matrices
  and the space of all Hermitian matrices which maps
  the Mueller matrices onto the positive semidefinite matrices.
  We use this to establish an optimality result for the filtering of Mueller matrices,
  which roughly says that it is always enough to filter the eigenvalues of the corresponding ``coherency matrix''.
  Then we further explain
  how the knowledge of the cone of Hermitian positive semidefinite matrices
  can be transferred to the cone of Mueller matrices with a special emphasis towards optimisation.
  In particular, we suggest that means of Mueller matrices should be
  computed within the corresponding Riemannian geometry.
\end{abstract}

\section{Introduction}

In polarisation optics Mueller matrices are of great importance,
as they describe the change of polarisation of light after interacting with a medium in a linear fashion.
In order to be a Mueller matrix the matrix has to  satisfy the Stokes criterion,
which states that every Stokes vector has to be mapped onto a Stokes vector.
Cloude then showed in \cite{10.1117/12.962889} that Mueller matrices can be associated
with Hermitian matrices with non-negative eigenvalues the so called coherency or covariance matrices.
This was then used for filtering measured matrices in order to make them physically meaningful,
i.e. satisfying the Stokes criterion.
Moreover, it was shown that any coherency matrix of a non-depolarising Mueller matrix
also known as a Stokes-Mueller matrix has only one non-zero eigenvalue.
This then easily suggests that any Mueller matrix is the sum of four non-depolarising matrices.
In \cite{kostinski1992depolarization} and \cite{le1996optical} matrices, which can be decomposed into a non-depolarising part
and one perfectly depolarising part, have been analysed and in the latter a filtering method was proposed.
Optimality of filtering was analysed in \cite{aiello2006maximum} by using a maximum likelihood method originally developed for quantum
process tomography and does such as \cite{faisan2013estimation} rely on
the Cholesky decomposition of the coherency matrix for filtering.
More about optimality filtering of Mueller matrices was derived in
\cite{boulvert2009decomposition}, \cite{Anna:11}, \cite{goudail2011polarimetric} and \cite{faisan2013estimation}. %
In \cite{Gil:16} then the optimality of the Cloude filter was rigorously proved.

The purpose of this work is to connect the methodologies of filtering of measured Mueller matrices
to well-established mathematical theories. We will show how this can be used to
prove a more general theorem about the optimality of filtering of Mueller matrices.
This simplifies and generalises part of the results of  \cite{Gil:16}.%
Moreover, we then review the mathematical theory about the Hermitian positive semidefinite cone
and explain, along with reviewing existing results, how this gives rise to
the differential geometry of the manifold of all Mueller matrices.

\section{Isometry of the ambient space}
In this section we explain how a well know result about the connection of Mueller matrices
and Hermitian positive definite matrices establishes an isometry between them.
For that we first restate the theorem which establishes this connection.
It implicitly first appeared in \cite{10.1117/12.962889}.
\begin{theorem}\label{sec:definition-3}(Theorem A.1 of \cite{van1993eigenvalue})
  Every Matrix $M\in \mathbb{R}^{4\times 4} $ with $M= (m_{ij})$ is a Mueller matrix
  if and only if the Hermitian matrix $H= (h_{ij})$ defined by the following linear equations has non-negative eigenvalues.
  Moreover, if the Mueller matrix has only one non-zero eigenvalue, then it is non-depolarising.
  \vspace{-.8cm}
  \small
  \begin{multicols}{2}
    \begin{equation}\label{eq:1}
      \begin{aligned}
        h_{00} =\frac{1}{2} (m_{00} +m_{11} + m_{22} +m_{33}),\\ 
        h_{11} =\frac{1}{2} (m_{00} +m_{11} - m_{22} -m_{33}),\\
        h_{22} =\frac{1}{2} (m_{00} -m_{11} + m_{22} -m_{33}),\\
        h_{33} =\frac{1}{2} (m_{00} -m_{11} - m_{22} +m_{33})
      \end{aligned}
    \end{equation}

    \begin{equation}\label{eq:2}
      \begin{aligned}
        h_{03} =\frac{1}{2} (m_{03} +m_{30} - Im_{12} +Im_{21}),\\
        h_{30} =\frac{1}{2} (m_{03} +m_{30} + Im_{12} -Im_{21}),\\
        h_{12} =\frac{1}{2} (m_{03} -Im_{30} + m_{12} +m_{21}),\\
        h_{21} =\frac{1}{2} (-m_{03} +Im_{30} + m_{12} +m_{21})
      \end{aligned}
    \end{equation}
  \end{multicols}

  \begin{multicols}{2}
    \begin{equation}\label{eq:3}
      \begin{aligned}
        h_{01} =\frac{1}{2} (m_{01} +m_{10} - Im_{23} +Im_{32}),\\
        h_{10} =\frac{1}{2} (m_{01} +m_{10} + Im_{23} -Im_{32}),\\
        h_{23} =\frac{1}{2} (m_{01} -Im_{10} + m_{23} +m_{32}),\\
        h_{32} =\frac{1}{2} (-m_{01} +Im_{10} + m_{23} +m_{32})
      \end{aligned}
    \end{equation}

    \begin{equation}\label{eq:4}
      \begin{aligned}
        h_{02} =\frac{1}{2} (m_{02} +m_{20} - Im_{13} +Im_{31}), \\
        h_{20} =\frac{1}{2} (m_{02} +m_{20} + Im_{13} -Im_{31}),\\
        h_{13} =\frac{1}{2} (m_{02} -Im_{20} + m_{13} +m_{31}),\\
        h_{31} =\frac{1}{2} (-m_{02} +Im_{20} + m_{13} +m_{31})
      \end{aligned}
    \end{equation}
  \end{multicols}
\end{theorem}
\normalsize
Note that we altered the result of these linear equations by a factor of $2$ in order to simplify oncoming observations,
but this does of course not change the validity of the theorem.

What to our knowledge has not yet discussed explicitly about the above result and the above equations is the following simple observation.
The whole trick is to realise that the Mueller matrices and Hermitian matrices are vectors
and then conclude that the Frobenius inner product coincide with the Hermitian/Euclidean inner product.

\begin{lemma}\label{sec:definition}
  Let $T$ be a linear map, which is defined by  the Equations \ref{eq:1}, \ref{eq:2}, \ref{eq:3} and \ref{eq:4}.
  The linear automorphism $T$ of the Hilbert space $\mathbb{C}^{4\times 4} $ with the usual Hermitian inner product
  is unitary.
  Moreover, the eigenvalues are $\{1, -1 \}$ with multiplicity $12$ and $4$ respectively and has therefore determinate $1$.
\end{lemma}
\begin{proof}
  We may assume that we are working in $\mathbb{C}^{16} $ by taking
  the canonical bijection from $\mathbb{C}^{4\times 4} $ to $\mathbb{C}^{16} $.
  It will be enough now to write down the complex $16\times 16$-matrix $T$ corresponding to the equations \ref{eq:1}, \ref{eq:2}, \ref{eq:3} and \ref{eq:4}.
  Compute the eigenvalues of $T$ with your favourite solver 
  and then
  conclude that $T^{\dagger}T=TT^{\dagger}=1$ where $T^{\dagger} $ is the conjugate transpose follows.
\end{proof}
\vspace{0.1cm}

The nice thing about unitary operators is that they preserve the Hermitian inner product, i.e.
we have that $ \langle x, y\rangle = \langle T(x),T(y) \rangle $ for any $x,y \in \mathbb{C}^{4\times 4} $.
Hence, the Hermitian norm (which coincides with the euclidean norm, in case there are only real entries) is preserved under these map.

Moreover, by Theorem \ref{sec:definition-3} we know that $T$ maps the
set of all Mueller matrices to the set of all positive semidefinite
Hermitian matrices.\footnote{Moreover, it might be of interest to some, 
  that we can define a Lie group structure on $4\times 4$-Hermitian positive semidefinite matrices corresponding 
  to the non-singular Mueller matrices by defining $A\cdot B= T(T^{-1}(A)T^{-1}(B)) $.} 
We further investigate some properties of the map
$T$.  We denote as $\mathcal C$ the $\mathbb{R} $-vector space of all
Hermitian $4\times 4$-matrices and denote as $\mathcal R$ the $\mathbb{R} $-vector space of
all $4\times 4$-$\mathbb{R} $-matrices both with the usual trace scalar product.

\begin{lemma}\label{sec:definition-1}
  The restriction $T\restriction \mathcal R$ of the linear map $T$ (as defined in Lemma \ref{sec:definition})
  is some non-singular orthogonal linear transformation 
  from $\mathcal R$ to $\mathcal{C} $. 
\end{lemma}
\begin{proof}
  We first consider $T$ to be a map from $\mathcal R$ to $4\times 8$-$\mathbb{R}$-matrices (map the complex numbers to $\mathbb{R}^{2}$).
  Moreover, the space of all Hermitian matrices can be considered a
  $16$-dimensional subspace of $\mathbb{R}^{4\times 8} $.
  The orthogonality and non-singularity follows  as the eigenvalues of the $T$ are $1, -1$ by Lemma \ref{sec:definition}.
\end{proof}
\vspace{0.1cm}

Now the next result follows by Theorem \ref{sec:definition-3}.

\begin{corollary}\label{sec:definition-4}
  The map $T$ is an isometry on $\mathbb{C}^{4\times 4} $ (with the Hermitian norm)
  and an isometry between  $\mathbb{R}^{4\times 4} $ and $\mathcal C$  (with the Euclidean norm)
  which maps the set of all Mueller matrices  onto the set of all semidefinite matrices. 
\end{corollary}

\section{Optimal filtering revisited}

Now we are able to translate the following problem into a question about Hermitian matrices: 
Given a real $4\times 4$-matrix (a Mueller matrix one got from a measurement).
Then we ask what the nearest (in terms of the euclidean distance) physically feasible  Mueller matrix is.
The same holds for the question, which asks for the nearest non-depolarising matrix to a given measurement.
Which now by Lemma \ref{sec:definition}, Lemma \ref{sec:definition-1} and Corollary \ref{sec:definition-4} can be translated to the question;
What is the nearest positive semidefinite matrix (with rank $1$ in the non-depolarising case)  to a given Hermitian matrix.
The answer to the first question by implicitly answering the second was already given in \cite{Gil:16},
but we can now rely on well-established mathematical theory to show this.
We will further derive a more general result and apply it to a further case.

\vspace{0.1cm}

\textit{Notation;} By $[a]$ we denote the diagonal matrix with entries $a_{n}\le\ldots \le a_{1} $ and
by $\mathcal U_{n} $ the set of all unitary $n\times n$-matrices.

\vspace{0.1cm}

The following is true in fact for any unitarily-invariant matrix norm $\norm{*}$
such as the Hermitian norm. It can be considered a 
Hermitian version of Theorem 4.5 of \cite{doi:10.1080/03081088708817747}.

\begin{theorem}\label{sec:definition-2}
  Let $c$ be some fixed real number.
  Let $Y$ be some non-empty closed subset of $$\{(x_{1},\ldots ,x_{n}) \in \mathbb{R}^{n}:x_{1}\ge \ldots \ge x_{n}\ge c \}.$$
  Further, let $S_{Y} $ be the set
  $$ \{V^{\dagger}[d]V: V\in \mathcal U_{n}, d\in Y \}.$$
  Given some Hermitian $A=U^{\dagger}[a]U$ with $a_{1}\ge \ldots \ge a_{n}\ge c$ and $U\in\mathcal U_{n} $,
  we then have that for some $b\in Y $ the following holds
  $$ \norm{A-U^{\dagger}[b]U}\le \norm{A-X} \quad \text{for all}  \quad X\in S_{Y} .$$
\end{theorem}
\begin{proof}
  The proof of Theorem 4.5 in
  \cite{doi:10.1080/03081088708817747} can be easily modified.
  Conclude in the same way that there exists some $b\in Y$ such that $\norm{A-U^{\dagger}[b]U}\le \norm{[a-x]}$ for any $x\in Y $.
  Since our matrices are Hermitian, we know then by Theorem 2 of  \cite{10.2307/2032661} that $\norm{[a]-[x]}\le \norm{A-X}$ for any $X\in S_{Y}$ which has $x$ as its eigenvalues.  
\end{proof}

\vspace{0.1cm}

This Theorem together with Corollary \ref{sec:definition-4} now lets us translate any nearness problems of Mueller matrices into a problem of nearness of the eigenvalues.

\begin{corollary}\label{sec:optim-filt-revis}
  Let $c$ be some fixed real number.
  Let $Y$ be a non-empty closed set in  $$\{(x_{1},\ldots ,x_{4}) \in \mathbb{R}^{n}:x_{1}\ge \ldots \ge x_{4}\ge c \}$$
  and let $M$ be a real $4\times 4$-Matrix such that the eigendecomposition of $T(M)$ is $U^{\dagger}[a]U$.
  Let $$ \mathcal{M}_{Y}=\{T^{-1}(V^{\dagger}[d]V): V\in \mathcal U_{n}, d\in Y \}.$$
  Then the nearest Mueller matrix in  $\mathcal{M}_{Y}$ in terms of the euclidean norm to $M$ is the matrix $T^{-1}(U^{\dagger}[d]U)$ 
  where $d\in Y$ is chosen such that $\norm{d-a}$ is minimal among all elements of $Y$. 
\end{corollary}

Hence, we can now easily conclude that the filtering proposed by Cloude \cite{10.1117/12.962889} is optimal.
For this let $M$ be the measured Matrix and let $T(M)=U^{\dagger}[a]U$ has a minimal eigenvalue of $c$.
We further assume that $c<0$ as otherwise we do not have to apply any filter.
Let $[a']$ be the tuple where we set all negative eigenvalues of $[a]$ to $0$.g
As the set $$\{(x_{1},\ldots ,x_{4}) \in \mathbb{R}^{4}:x_{1}\ge \ldots \ge x_{4}\ge 0 \}$$
is closed in $$\{(x_{1},\ldots ,x_{4}) \in \mathbb{R}^{4}:x_{1}\ge \ldots \ge x_{4}\ge c \} ,$$
Corollary \ref{sec:optim-filt-revis} lets us conclude that $T(U^{\dagger}[a']U)$
is the nearest Mueller matrix estimate of $M$.
In similar fashion we can conclude that setting all but the biggest eigenvalue of $T(M)$ to $0$
will give us the best estimate for non-depolarising Mueller matrices.

Let us now consider all Mueller matrices $M$ which can be decomposed as a sum of a non-depolarising part $P$ and a perfectly depolarising matrix $D$,
i.e.~$D$'s only non-zero element is the upper left entry.
Now if we map $D$ via $T$ onto the Hermitian matrices we will see that $T(D)$ is a diagonal matrix $[(d\ldots d)]$.
We continue by stating some simplified version of Weyl's inequality.
\begin{theorem}
  Let $A,B, C$ be Hermitian matrices.
  If $A + B =C$ and $a_{n}\le \ldots \le a_{1} $, $b_{n}\le \ldots \le b_{1} $ and $c_{n}\le \ldots \le c_{1} $ be their eigenvalues,
  then we have that $a_{i}+b_{n}\le c_{i}\le a_{i}+b_{1} $ for all $1\le i\le n$.
\end{theorem}
Now in our case, if we take $T(P)=A$, $T(D)=B$ and $T(M)=C$ and let $p_{i} $ be the eigenvalues of $T(P)$
and $x_{i} $ be the eigenvalues of $T(M)$,
we get that $p_{i}+d=x_{i} $ for all $i$ where $1\le i\le 4$ holds.
\footnote{Of course, it is not logically necessary to apply Weyl's inequality here
  as $B$ is diagonal under any basis. But we do for the sake of introducing longstanding mathematical results.}
Hence the matrix $T (M)$ has eigenvalues of the form $x_{1}\ge x_{2} =x_{3} =x_{4}\ge 0$.
On the other hand, take any Matrix $H$ which has eigenvalues of this form.
Then subtracting $[x_{4},\ldots,x_{4}]$ will give us by applying Weyl's inequality again,
that $C = H -[x_{4},\ldots,x_{4}]$ has eigenvalues $x_{1}-x_{4}, 0, 0 ,0$ and therefore $T^{-1} (C)$ is non-depolarising.
Hence we know that $T^{-1} (S_{E})$ with $E=\{(x_{1},x_{2},x_{3},x_{4}) \in \mathbb{R}^{4}:x_{1}\ge x_{2}=x_{3}=x_{4}\ge 0\} $ 
is the set of all Mueller matrices which can be decomposed in a non-depolarising part
and perfectly depolarising part.\footnote{This was already mentioned in \cite{ossikovski2008depolarizing}} %
Now asking what the best estimate for a measured matrix, which has this type of composition, can be answered by
applying Corollary \ref{sec:optim-filt-revis} with ${E} $ (which is closed).
Now if $a_{1}\ge a_{2} \ge a_{3} \ge a_{4} $ are the eigenvalues of $T (M)$ then
$b= (a_{1},c,c,c)$ with $c =\frac{1}{3}\sum_{i= 2}^{4}a_{i} $ is the best estimate in $E$.
And hence, we have that for a measurement $M$ the best estimate in $T^{-1} (S_{E})$ is $T^{-1} (U^{\dagger} [b]U) $.
This also shows us that Equation (17) of \cite{le1996optical} is in fact the best \textit{a priori} estimate  for the  perfectly depolarising part,
contrary to what was stated in that paper.

\section{Geometry of the semidefinite cone}

The reader may wonder why we could so easily compute the nearest Mueller matrix to a given real matrix
or respectively solve the corresponding problem the nearest semidefinite matrix to Hermitian.
This ultimately has to do with the nature of the object consisting of all complex semidefinite matrices.
It turns out that this is a deeply studied object which is known under the name \textit{complex semidefinite cone} 
or more generally symmetric cones and is used among other things for \textit{complex semidefinite programming.}
Many very nice properties such as convexity are known about it.
In fact, it is a cone. So it is closed under positive linear combinations,
i.e. $\alpha H_{1} +\beta H_{2} $ is also positive definite
with $H_{1},H_{2} $ positive definite and $\alpha,\beta $ positive numbers.
Of course, this implies that the set of all Mueller matrices is a cone by linearity of $T $.
It is also known what the interior (all positive definite matrices) and  the boundary (all singular positive semidefinite matrices) is.
Moreover, there is a Riemannian metric tensor on its interior (see \cite{hill1987cone} and Chapter 6 of \cite{book:71688}).%

Now the practitioner can use this knowledge and grab ready available tools and mathematical theory.
For example, take a subset of Mueller matrices $S$ and a function $f:S\to \mathbb{R} $ one wants to optimise.
We have just seen such functions namely the distance of the Mueller matrices (or certain subsets of them) to a given measurement $M$.
As done before, we can translate the problem by optimising the map $f\circ T^{-1}$ from $T(S)$ to $\mathbb{R} $ instead.
This can be either done by finding suitable theory about the semidefinite cone such as Theorem \ref{sec:definition-2}
and then solve the problem directly.
Or a more general approach would be to use available tools for solving optimisation problems.
As a start one would transfer the complex optimisation problem into a real one (with tools as YALMIP \cite{Lofberg2004}).
Although voices have been raced to consider optimisation in the complex numbers directly \cite{gilbert2017plea}.
In any way, there are many available software tools for computing the optimum of a function on the complex or real semidefinite cone
such as Manopt \cite{manopt}, Pymanopt \cite{JMLR:v17:16-177} and SeDuMi \cite{sturm1999using}.

We highlight one approach of characterising the space of semidefinite matrices of some fixed rank
taken from \cite{6638382} and \cite{vandereycken2009embedded}
which is also described in the code of \cite{manopt} and \cite{JMLR:v17:16-177}.
Now if the rank is $1$ then this space is in correspondence via $T$ with the non-depolarising Mueller matrices.
The differential geometry of the non-depolarising Mueller matrices was already studied in \cite{article2010Devlaminck}.
We going to outline now the differential geometry of the Hermitian positive semidefinite cone. 

A semidefinite matrix $H$ from $\mathbb{C}^{4\times 4} $ of rank $k$ can be written as an outer product $YY^{\dagger} $
of a matrix $Y$ of $\mathbb{C}^{4\times k} $ of full rank.
On the other hand any such outer product  $Y Y^{\dagger} $ is positive semidefinite and of rank $k$.\footnote{The same factorisation was already used in \cite{sheppard2018factorization}, although they did not consider the subtleties of the rank and the oncoming uniqueness properties.}
As in \cite{6638382} we define an equivalence relation on $\mathbb{C}^{4\times k} $
by identifying $YU$ with $Y$ for all unitary matrices $U$ (as the outer product does not change, i.e. $YY^{\dagger} =YU(YU)^{\dagger}$).
We denote the manifold of all $\mathbb{C}^{4\times k} $ matrices of full rank as ${C_{4k}} $.
Now by the quotient manifold theorem the manifold ${C}_{4k}/{U}(k) $ is a Riemann quotient manifold, if ${U}(k) $
is the Lie group of all unitary matrices.

One can note a striking similarity to Cholesky decomposition, which is used in \cite{aiello2006maximum} and \cite{faisan2013estimation}.
In particular, in case of positive definite matrices the Cholesky decomposition is unique
and ${C}_{44} $ can be replaced with all triangular matrices with real diagonal entries.
In the case $k<4 $ one can find a unique decomposition after a twisting with permutation matrices \cite{higham1990analysis}
and hence one would end up with a finite-to-one map (bounded by 24, the number of $4\times 4 $-permutation matrices).
For all $k$ the metric of the manifold is given by the real-trace inner product,
if identifying the complex numbers with $\mathbb{R}^{2} $.
Moreover, when $k=1$ we can find a representative of the equivalence classes
by requiring that the first non-zero element of the tuple $c \in \mathbb{C}^{4} $ is a real number.
This lets us conclude that its dimension is $7$.

Further, if we identify the $\mathbb{C} $ with $\mathbb{R}^{2} $
then the quotient manifold theorem tells us also that the dimension of the Riemannian manifold
of all complex positive semi-definite matrices of rank $k$ has dimension $4\cdot k-k^{2} $.
Of course, all this analysis extends to the Mueller matrices by extending the mapping via $T^{-1} $.
This then implies that the manifold of Mueller matrices is a decomposition
of this quotient manifolds ${C}_{4k}/{U}(k) $ with $1\le k \le 4$
and the zero element.
Furthermore, in case where the Mueller matrices $M$ are assumed to be the sum of a non-depolarising matrix and an ideal depolariser
and hence the corresponding coherency matrices $T(M) $ are a sum of a rank-$1$ positive semidefinite matrix
and a diagonal matrix with positive entries,
it is not hard to see that this manifold is the product manifold of the positive real numbers $ \mathbb{R}_{+} $ 
and the manifold of all complex rank-$ 1 $ positive semidefinite matrices.

What also can be interfered from the above analysis is the following.
We set the above together to receive a map $F$ which is defined as follows
\begin{equation}
  \mathbb{R}^{2\times 4\times k}\to_{\mathbb{R}^{2}= \mathbb{C}} \mathbb{C}^{4\times k}\to_{YY^{\dagger}}\text{HPSD}\to_{T^{-1}} \mathcal{M}     
\end{equation}
where \textit{HPSD} is the space of all Hermitian positive semidefinite matrices and $\mathcal{M} $ the space of the Mueller matrices.
Now a short calculation gives us then that $F$ is a quadratic homogeneous polynomial and hence any $F(\lambda x)=\lambda^{2} F(x) $.
Moreover, we can see that $\norm{x}_{\text{Euclidean}}^{2}=F(x)_{11} $, where $F(x)_{11}$ is
the upper left element of the Mueller matrix.
This means that it is almost always enough to study the reduced case of Mueller matrices 
which have upper left element $1$.

Another question which now arises is that of the mean of two or more matrices.
In the euclidean space this of course just the standard Arithmetic mean.
But in manifolds the geodesic might look very different from a straight line
and hence the average of two matrices, i.e. the middle point on the geodesic between these two,
might be significantly different from the arithmetic mean.
This case of the geometric average of two Mueller matrices was already covered in \cite{devlaminck2010mueller}.
The generalisation of this concept namely the Riemannian barycenter of matrices $A_{1}\ldots A_{n} $,
i.e.~the matrix  which is the minimum of the function $\sum_{i=1}^{n}d(X,A_{i}) $ where $d$ is the distance measure on the manifold.
Again we can rely on a well studied area of means of semidefinite linear operators. 
Studying of the mean of two linear operators began through a study of connections of electrical networks \cite{anderson1969series}.
This was then followed by more axiomatic studies on general Hermitian operators \cite{kubo1980means}, \cite{pusz1975functional}.
Means between more than two matrices have been studied in \cite{ando2004geometric}. 
In \cite{bonnabel2009riemannian} means have been studied in case of real  semidefinite matrices of fixed rank.
An exposition of the geometric nature of means can be found in Chapter 6 of \cite{book:71688}.
All together this suggests that computing the mean of multiple Mueller matrices 
should be done using the Riemannian geometric mean. 
In practice this would be done by transferring them via $T$ to the semidefinite cone
and then using available implementation of the Riemannian mean such as the tool Yalmip \cite{Lofberg2004}. 

\section{Conclusion}
We have established a connection between the area of Mueller matrices 
and the areas of general matrix analysis, Riemannian geometry and optimisation.
All basically by interpreting existing results
and making the simple observation that the real ambient space of the Hermitian positive semidefinite matrices
and the Mueller matrices and the objects themself isometrically map onto each other. 
With this new knowledge, we showed 
how matrix analysis can be directly used
to prove an optimality result (see Corollary \ref{sec:optim-filt-revis})
for the filtering of measured Mueller matrices.

We further reviewed mathematical results about the complex semidefinite cone
and noted how this can be used with our previous results and how this suggests a new mean for Mueller matrices.
Of course, such connection have been partly discovered in the past or 
general results about semidefinite matrices have been reproved in the special case of Mueller matrices 
and $4\times 4 $-Hermitian semidefinite matrices.
But our connection makes this precise and provides a way to bring well-established mathematical theories and tools 
into the polarimetric world.
One can also speculate that the analysis which we have established here, 
might bring new insight to  quantum optics and quantum information 
as they share some mathematical objects \cite{aiello2006maximum}.

What is still missing in our analysis is to bring together this analysis 
with the study of the Lie group structure of invertible Mueller matrices. Or more generally the semigroup structure.
Of course, by our analysis of the geometry it is easy now to compute the tangent space at the identity and therefore the Lie algebra.
But this is nothing new, the study of the Lie group and Lie algebra was already done in  \cite{devlaminck2008definition}.
What is still missing is a study how the geometry of the additive structure of the Mueller matrix,
which corresponds to parallel optical elements,
and the geometry of the multiplicative structure, which corresponds to successive optical elements, interact.

\printbibliography

\end{document}